\documentclass{article}
\usepackage[utf8]{inputenc}
\usepackage{graphicx}  
\usepackage{dcolumn}   
\usepackage{bm}        
\usepackage{amssymb}   
\usepackage{amsmath}   
\usepackage{mathrsfs}
\usepackage[normalem]{ulem}
\usepackage{color}
\usepackage{empheq}
\usepackage{hyperref}
\usepackage[modulo]{lineno}
\usepackage{float}
\usepackage{soul}
\usepackage{authblk}
\usepackage[
backend=biber,
style=alphabetic,
sorting=ynt
]{biblatex}
\usepackage{amsthm}
\newtheorem{theorem}{Theorem}
\newtheorem{proposition}{Proposition}
\newtheorem{remark}{Remark}

\newtheorem{definition}{Definition}

\newcommand{\comment}[1]{{{#1}}} 

\title{Limit Order Book (LOB) shape modeling in presence of heterogeneously informed market participants}

\author[1]{Mouhamad DRAME}
\affil[1]{FiQuant, Laboratoire de Mathematiques et Informatique pour la Complexite et les Systemes, CentraleSupelec, Universite Paris-Saclay,3 rue Joliot Curie, 91190 Gif-sur-Yvette}
\begin{document}

\maketitle

\begin{abstract}

The modeling of the limit order book is directly related to the assumptions on the behavior of real market participants. This paper is twofold. We first present empirical findings that lay the ground for two improvements to these models.The first one is concerned with market participants by adding the additional dimension of informed market makers, whereas the second, and maybe more original one, addresses the race in the book between informed traders and informed market makers leading to different shapes of the order book.
\\
Namely we build an agent-based model for the order book with four types of market participants: informed trader, noise trader, informed market makers and noise market makers. We build our model based on the Glosten-Milgrom approach and the most recent Huang-Rosenbaum-Saliba approach. We introduce a parameter capturing the race between informed liquidity traders and suppliers after a new information on the fundamental value of the asset. We then derive the whole `static" limit order book and its characteristics -namely the bid-ask spread and volumes available at each level price- from the interactions between the agents and compare it with the pre-existing model. We then discuss the case where noise traders have an impact on the fundamental value of the asset and extend the model to take into account many kinds of informed market makers. 
\end{abstract}

\tableofcontents

\section{Introduction}

Most modern financial markets are order-driven markets, in which all of the market participants display the price at which they wish to buy or sell a traded security, as well as the desired quantity. This model is widely adopted for stocks and futures, due to its superior transparency.

A market for a security is 'liquid' if investors can buy or sell large amounts of the security at a low transaction cost. Liquidity is a valuable characteristic of an asset because it allows investors to trade at prices close to their decision price. Liquidity is generally supplied by market makers, who are willing to take the other side of a trade for a premium relative to the current fundamental value. Traders, who are willing to pay this premium in order to execute a trade immediately, demand liquidity.\\

In an order-driven market, all the standing buy and sell orders are centralized in the limit order book (LOB). Orders in the LOB are generally prioritized according to price and then to time according to a FIFO (First-In First-Out) rule. Some exchanges prioritize according to a proportional rule and others use a mix of these two rules.

With these electronic markets, trading strategies have become more and more important. In
particular, market making - liquidity providing - strategies lay at the core of modern markets. Since 
there are no more designed market makers, every market participant can\footnote{sometimes they must post 
limit orders for tactical reasons} provide liquidity to the market, and the choice of quantity of shares 
to add -otherwise the LOB shape- is a question of crucial practical relevance.\\

A trader may desire to trade immediately because he has some private information about the future value of the
asset or because he wants to rebalance his portfolio. The presence of traders with private information 
exposes the market makers to adverse selection risk and consequently impacts the prices they quote. 
The specific trading rules of the exchange also impacts the premium that the trader pays for liquidity. 
The interaction between trader information, the market makers and the trading rules is at the heart of 
many policy questions to improve market quality. In order to answer questions about if large trading 
costs (i.e high spread and/or small market depth) are due to adverse selection costs or strategic 
market makers, it is necessary to consider models that can incorporate these effects.\\

Our starting point is to assess the empirical literature considering a generic group of market makers 
with homogeneous level of information. Several models in the literature study the LOB assuming the presence of
one type of market makers -while it is generally admitted that asymmetry of information exist 
between speculators traders (considered informed) and pure liquidity traders (considered uninformed). 
For instance the papers \protect\cite{BaruchGlosten2019, huang2019glosten,glosten1985bid} consider all identical liquidity suppliers with no asymmetry of information between them.\\

Our paper differs from these by introducing an asymmetry of information between liquidity suppliers. We empirically bring evidence of this heterogeneous level of information. We consider for this purpose the ``trade signature'', a metric to empirically assess the level of information of both liquidity traders and liquidity providers. We identify different clusters of liquidity traders and liquidity providers based on different metrics and build on this an agent-based model to derive the whole LOB shape.\\

Our model shares therefore multiple similarities with the existing models by using the conditional tail expectation \protect\cite{glosten1994electronic}, which is in our case (where liquidity suppliers face informed traders and have some latency with respect to these informed traders) a zero-profit condition for an order added at the back of a queue.\\

After the preamble on the nature of market participants, we build our agent-based model as in \protect\cite{huang2019glosten, glosten1985bid}. Contrary to these papers we consider 4 types of market participants: an informed trader, a noise trader, informed market makers and noise market makers. The informed trader and the informed market maker receive information on the efficient price and try to take advantage of this knowledge by respectively sending a market order and a cancel order (with a given probability that the informed market maker order is executed first). The noise market makers receive this information with some delay. The noise trader sends order in a zero intelligence way. We first consider the case when the noise trader does not influence the efficient price and later discuss the case where noise traders have some toxicity for liquidity providers.\\

We analyze how adverse selection, competition between market makers and liquidity traders affect the liquidity in this pure limit order market. In this type of market, the cost of liquidity at any point in time is determined by the LOB, which contains all outstanding limit orders. A market maker's decision to submit a limit order to the LOB involves a trade-off between the premium and the adverse selection he faces.

We deduce from this framework the link between the equilibrium state of the LOB, the different intensities of the dynamics of the efficient price move, the race to order insertion parameter between informed market makers and informed traders, the bid-ask spread. The bid-ask spread emerges as the result of the negative expected returns of limit orders placed too close to the efficient price. Most interestingly we find that due to the presence of informed market makers the bid-ask spread reduces.\\

The effect of market composition on spread is as follows. When informed traders are faster than market makers(all things being equal), the spread takes a non-zero value as market makers post their order far to avoid adverse selection. Market makers therefore stop providing liquidity when the fraction of informed traders comprising the market becomes sufficiently large. But when informed market makers are faster than informed traders, the spread becomes smaller and the liquidity increase due to the surplus of liquidity from informed market makers that will be less adverse selected. In the case informed market makers are sufficiently faster than informed traders we find that the spread vanishes and the liquidity increases (this is true if all the informed market makers are homogeneously informed but as soon as there is an asymmetry between them, the liquidity decreases).\\

The rest of this paper is organized as follows. In section 2 we do some empirical studies to cluster liquidity traders and liquidity suppliers and present the considered data-set. In section 3 we describe our model with the different market participants. We define a race parameter which takes into account the race between the informed market makers and the informed trader. We solve the model in the case the noise traders do not have any toxicity for the market makers (by first solving the non-zero tick size case and then the the case with tick size). We give different results on the shape of the book and demonstrate that the noise market makers and the informed ones have different optimal shapes. We derive the bid-ask spreads in this case and compare it to the case when the informed trader is always faster than the informed market makers. We then show how the model can be extended to the case where noise traders have some toxicity for the liquidity suppliers. In section 4 we extend the model to take into account many types of informed makers and discuss how this influences market quality. In the last section, we discuss our modeling choice before concluding.
\\

In all the following bold uppercase symbols denote matrices, bold lower cases denote vectors and light lower cases denote scalars.

\section{How informed are market participants?}
Modern markets are two-sided financial markets with market makers quoting bid and ask prices on one side of the market and traders submitting market orders on the other side.

The market makers buy low, and sell high, adjusting their bid-ask spreads in accordance with the
adverse selection\footnote{buying when the price is going down or selling when it is going up, i.e the correlation between trade and delta price} they face. Traders, on the other hand, benefit from market makers competing to offer the best quotes. The set of traders comprises informed traders (speculators) as well as noise traders (liquidity traders who trade for reasons due to liquidity shocks unrelated to the asset value).

Since Kyle, and Glosten and Milgrom \protect\cite{glosten1985bid,kyle1985continuous}, a common assumption in market microstructure literature on asymmetric information is that while market-makers do not have superior information on market fundamentals, some traders have private information. The canonical model of dealer markets is due to Glosten and Milgrom \protect\cite{glosten1985bid}, who assume that traders are better informed than market makers.

Empirical evidence show that market makers can be better informed than some traders in the information acquisition game. \protect\cite{moinas2010hidden} analyzes potentially informed liquidity provision with hidden limit orders and shows that informed market makers act on dark pools to benefit from their superior information, \protect\cite{biais1993price} who considers market makers with private information about inventories.

In the following, we present some statistical properties on market participants showing heterogeneously level of information for both liquidity takers and liquidity providers.

\subsection{Empirical clustering of market participants}

\subsubsection{Trade signature}
We define the trade signature at horizon k as:
 $$ST(k)=\epsilon \frac{\sum_{t} Q_{t}\left(X_{t+k}-P_{t}\right)}{\sum_{t}\left|Q_{t}\right|}$$
where 
\begin{itemize}
    \item $X_t$ denotes the asset's efficient price at time $t$
    \item $P_t$ denotes the effective trade price at time $t$
    \item $Q_t$ denotes the trade quantity at time $t$ (positive if the trade triggered by the liquidity taker is a buy, negative otherwise)
    \item $\epsilon$ is equal to 1 for trade signature of aggressive traders and -1 for liquidity providers.
\end{itemize}

The rationale behind the trade signature is that informed traders should win on average and thus have their signature positive at more or less short horizons while noise traders should lose on average and have their signature negative at intraday horizons. The trade signature is actually a normalized $P\&L$ for market participants with respect to their efficient price.

In the following we consider (unless formally specified) the micro-price as the efficient price. The micro-price is defined as 
$$
m p_{t}=\frac{b i d_{t} V_{t}^{a}+a s k_{t} V_{t}^{b}}{V_{t}^{a}+V_{t}^{b}}
$$

where $bid_t$ (resp $ask_t$) denotes the bid (resp ask) price at time $t$ and $V_t^b$ (resp $V_t^a$) the volume available at the best bid (resp ask) limit in the book. 
We define also the mid-price as 
$$
mid_{t}=\frac{bid_{t}+ask_{t}}{2}
$$
The microprice has nice properties compared to the mid-price (for example the intuitive idea that the efficient price is closer to the side of the book where there is less volume is well taken into account by the microprice).

We will mark all the signature with respect to the micro price and to the mid-price in order to avoid any artifact due to liquidity at best limits.

\subsubsection{Clustering methodologies and their rationale}

We cluster liquidity providers and liquidity takers based on many criteria related to order deposition time, order modification, order volumes and order posted price distance to efficient price.
In the following we will denote passive traders by PT and aggressive traders by AT.

\paragraph{Clustering based on time}
The idea of clustering traders based on the time they send their orders is natural as one expect
that informed traders will try to profit from their superior information as soon as they have a new information. Thus they should send market orders as soon as they see a new update in the efficient price. In the same way informed market makers are expected to send their limit orders soon to gain good queue positions.

\subparagraph{Clustering Market Makers} 
We use many time metrics to cluster market makers.
\begin{enumerate}
    \item \textbf{Trade to add time:} We consider a market maker as informed (conditional on his order 
    being executed) if the duration between the last trade before he adds his order and the moment he adds 
    his order is lower than a given threshold and otherwise uninformed. We actually generalize this 
    clustering method to take into account heterogeneously level of information for market makers. Given 
    the thresholds $th_{1} < ... < th_{n-1}$ we create $n$ clusters of passive 
    traders $PT_{0}$, ..., $PT_{n-1}$ where $PT_{0}$ is the best informed passive trader and $PT_{n-1}$ the 
    less informed one. The rationale is that informed PT will act quickly after a trade to add quotes.
    
    \item \textbf{Add to add time:} Another criteria to cluster liquidity provider is to consider the 
    duration between the last add time of an order at a level price and the moment he adds at the same 
    level price (conditional on the order being executed). We expect that uninformed traders add orders at 
    a given price far from the last add time at the same level price, while informed may add in a 
    clustered way.
\end{enumerate}

\subparagraph{Clustering Aggressive Traders}
We cluster aggressive traders based on the duration between the time they send their trade and
the last trade time before their orders execution based on the idea that informed traders will rush for liquidity at the same time buckets (they generally use the same sources of information).
Given the thresholds $th_{1} < ... < th_{n-1}$ we create n clusters of aggressive traders $AT_{0}$, ..., $AT_{n-1}$ where $AT_{0}$ is the best informed aggressive trader and $AT_{n-1}$ the less informed one.

\paragraph{Clustering based on trade volumes}
We cluster AT based on the liquidity they take in the book. We assume that informed AT will send
market orders that take a high proportion or deplete the limit they trade with (or take exactly the quantity available at the best limit they consume). The criteria we consider is thus the ratio between the traded volume and the available volume at the considered limit before the trade. Informed traders are supposed to send market orders that consume exactly the best limit -in a pure utilitarian point of view as they are informed on the next efficient price jump they have interest of sending orders that consume all the available liquidity up to the jump and thus deplete exactly a limit order (or more if the jump size is greater than the first limit price).

\paragraph{Clustering based on updates}
We cluster LP based on the number of updates (modifications of price or size) of their order between
the time they add it and the time it is executed. The motivation is that informed PT may update their orders to incorporate new information.

\subsection{Data Overview}

Our attention is restricted to US treasury bonds futures traded on CME exchange. We consider 5 US treasury bonds futures, namely UB(Ultra US Treasury Bond), ZB(US Treasury Bond), ZF(5-Year T-Note), ZN(10-Year T-Note) and ZT(2-Year T-Note) traded on CME for a period of 144 trading days starting from June 04,2018 to Jan 02,2019. The trading hours (in US central time) go from Sunday to Friday each day between 5:00 pm to 4:00 pm (with a 60-minute break each day beginning at 4:00 pm). We used the so-called market-by-order data feed as opposed to the market-by-level data. The market-by-order contains all order book events including limit order postings, trades, and limit order cancellations with their given ids. Market-by-order makes it possible to follow each order between the time it is added in the book and the time it is executed or canceled with all the modifications of the order. Thus it is easy to compute all the clustering classifications for each order executed and to compute easily trade signature by clusters. It also allow to easily compute the realized gain of orders with respect to the efficient price and to compare it with the theoretical gains as computed by the market makers prior to adding them (conditional on their execution right now).

\begin{figure}[H]
\centering
 \includegraphics[width=\textwidth]{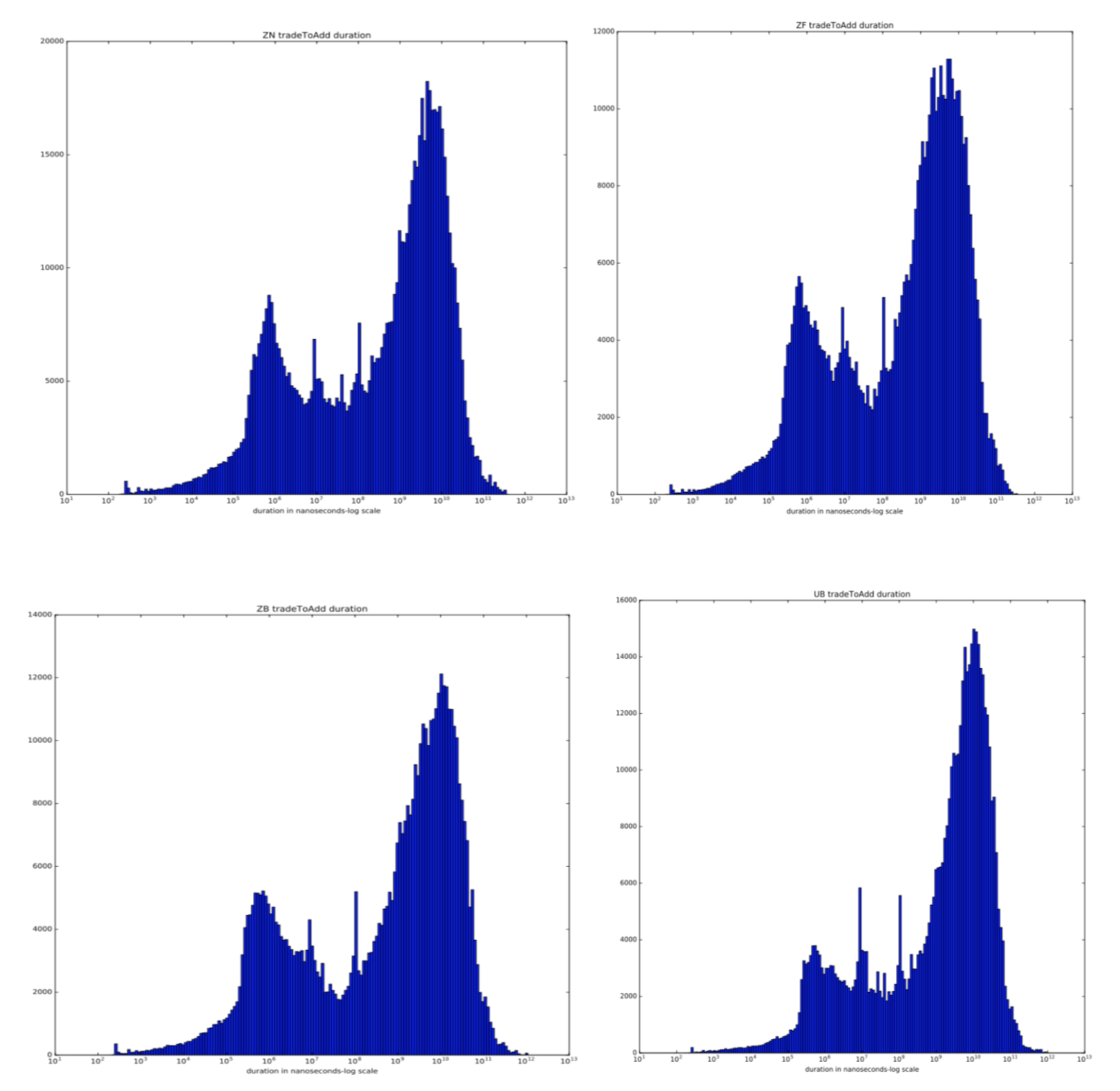}
 \caption{\textbf{Histograms for duration between last trade event and new add event for different US Treasuries Futures on 27/11/2018.}\newline
 Top left: ZN, top right ZF, bottom left: ZB, bottom right: UB
 The characteristic peaks are consistent across days.}
 \label{fig:tradetoadd}
\end{figure}

\begin{figure}[H]
\centering
 \includegraphics[width=\textwidth]{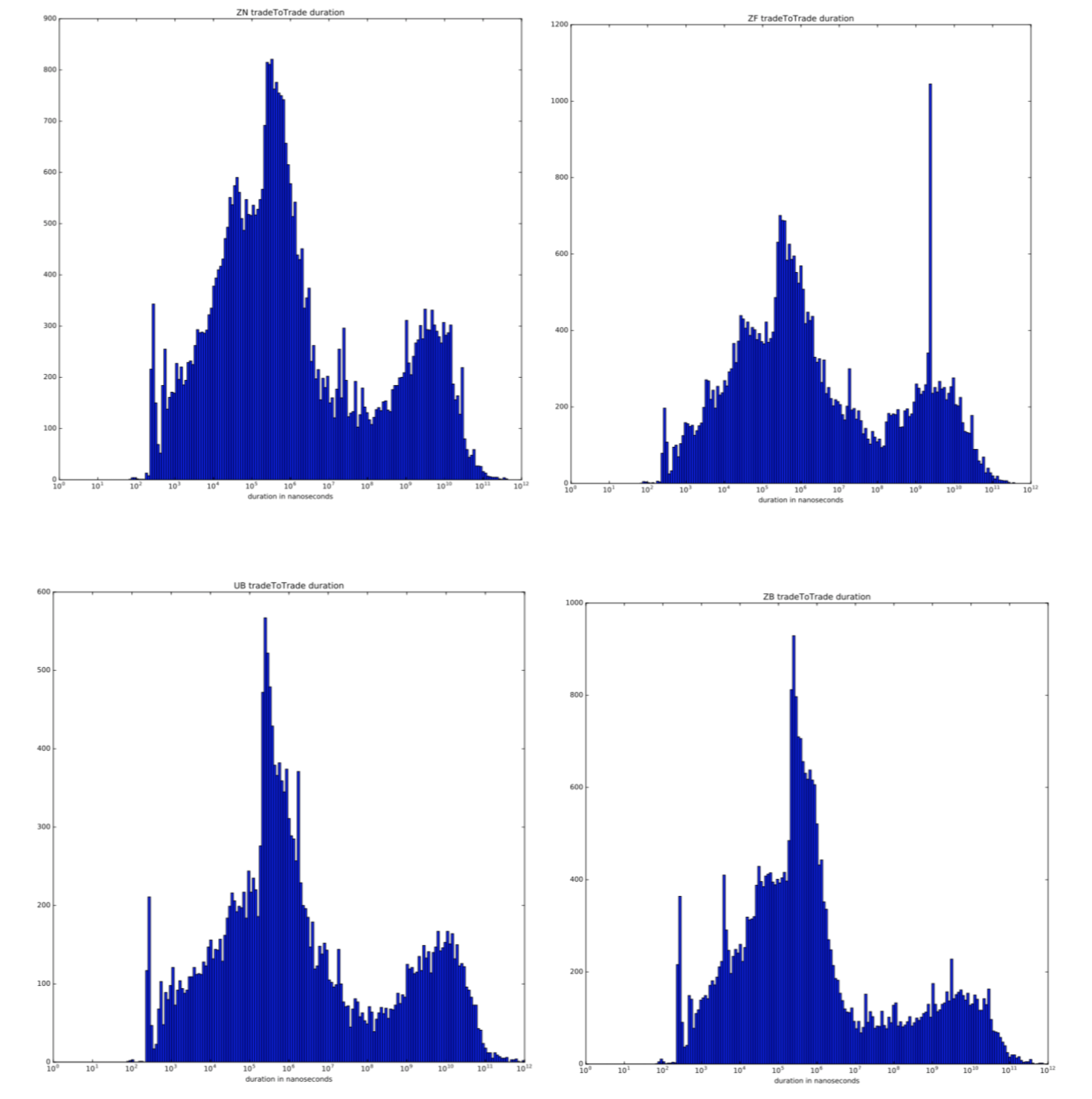}
 \caption{\textbf{Histograms for duration between last trade event and the last trade before on the same product for different US Treasuries Futures on 27/11/2018.}\newline
 Top left: ZN, top right ZF, bottom left: ZB, bottom right: UB \newline
 The characteristic peaks are consistent across days.}
 \label{fig:tradetotrade}
\end{figure}

\newpage

\subsection{Analysis of trade signatures by market participants groups}
We present some results on the clustering based on time and based on volumes in this part (the other results are relegated in appendix \eqref{app:graphs}).\\

We precise that the different clustering are standalone clustering where it is not necessary to have the type of the order flow categories of the agents (such as retail flow, high frequency traders, etc). The authors of \protect\cite{megarbane2017behavior} for example use a classification based on data set provided by the French regulator ``Autorité des Marchés Financiers". So we do not need to have the so-called ``client id".
Our different classifications are $a$ $posteriori$ clustering but one can for clustering depending on time do it online with appropriate market data (market by order data with order ids and all details of an order). \\ 

We compute the different trade signatures at time horizons up to 1000 seconds.\\

In Fig.~\ref{fig:signaturetradetoadd} and Fig.~\ref{fig:signaturevolume} we present the results with a clustering in 2 groups for AT and for PT. \\

As expected at the time of trade the trade signature for aggressive orders (resp passive orders) is 
negative (resp positive) due to the basic fact that liquidity takers start by losing money by crossing the spread. 

The trade signature remains constant on horizon up to more or less 100 ns depending on the considered asset due to the fact that the efficient price does not move in this interval just after a trade. Then it increases (resp decreases) and becomes (resp stays) positive for informed aggressive traders (resp informed passive traders) while it stays negative for uninformed aggressive traders (informed and uninformed depend on our classification) and this is not perfectly symmetric for liquidity providers and liquidity takers as our classification does not enforce to have an informed passive trading against an informed aggressive.

The different metrics we define allow a separation in different clusters based on market participants realized $P\&L$.

It is therefore important to consider that a market maker can be informed and to incorporate this in LOB modeling. The results in appendix \eqref{app:graphs} suggest also that there are different level of information between informed market makers (the ones with positive signatures) and so one also should take this into account when modeling the order book. 

\begin{figure}[H]
\centering
 \includegraphics[width=\textwidth]{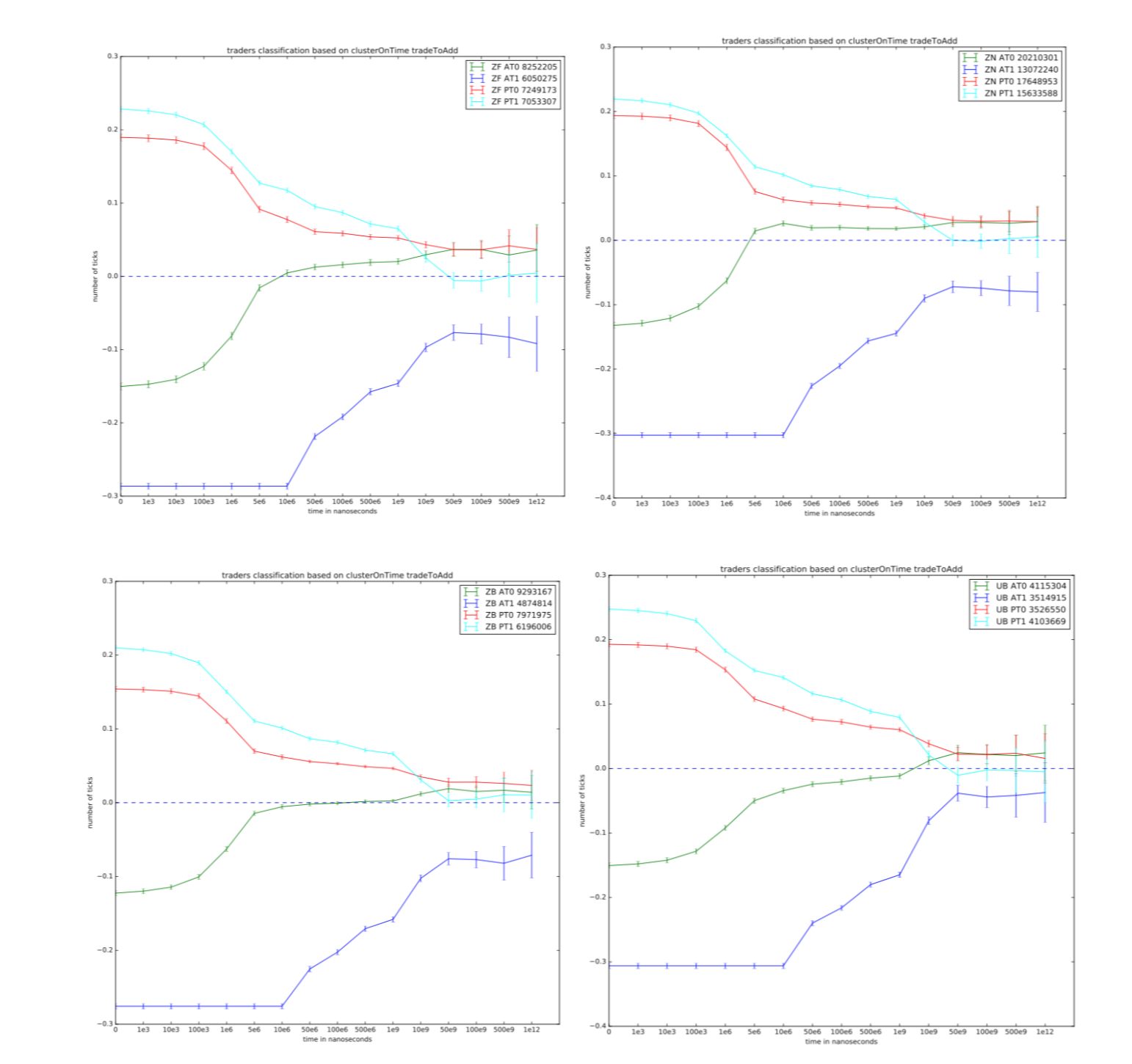}
 \caption{\textbf{Trade signatures for different products based on trade to add time clustering for passive traders and trade to trade time for aggressive traders.}\newline
 The numbers indicated correspond to the number of points in each cluster during the considered period (144 days) \newline
 Top left: ZF, top right ZN, bottom left: ZB, bottom right: UB \newline
 We use a threshold of $10^{7}ns$ for both PT and AT.}
 \label{fig:signaturetradetoadd}
\end{figure}

\begin{figure}[H]
\centering
 \includegraphics[width=\textwidth]{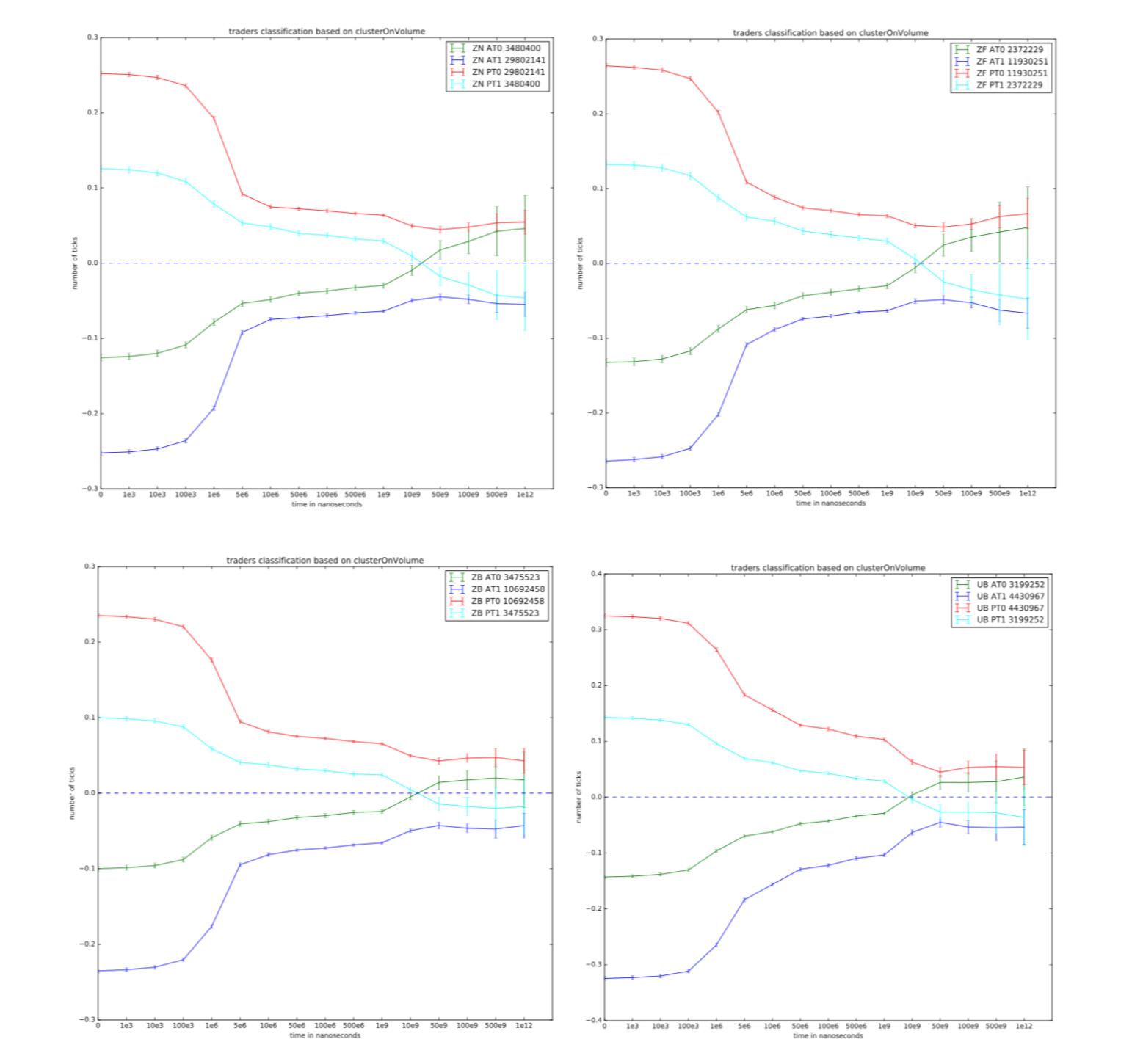}
 \caption{\textbf{Trade signatures for different products based on volume clustering}\newline
 The numbers indicated correspond to the number of points in each cluster during the considered period (144 days) \newline
 Top left: ZN, top right ZF, bottom left: ZB, bottom right: UB \newline
 We use a threshold of 0.1 of the ratio between traded volume and available volume at best limit at the time of trade to separate the two groups}
 \label{fig:signaturevolume}
\end{figure}

\section{The model}

\subsection{Efficient price dynamics}
We consider an asset whose underlying efficient price is composed of a jump component with Poisson arrival times and a noise trader component as follows:
$$
P(t)=P_{0}+\sum_{l=1}^{N_{t}} B_{l}+\sum_{j=1}^{N_{t}^{\prime}} \theta\left(Q_{j}^{u}\right)\left(X_{j}-\rho X_{j-1}\right)
$$
where $P_0 > 0$ and the innovation part is composed of 2 terms:
\begin{itemize}
    \item $\sum_{l=1}^{N_{t}} B_{l}$ is a compound Poisson process and $(N_{t})_{t \geq 0}$ is a Poisson process
    with intensity $\lambda_i$ and l=1 the $(B_{l})_{l>}0$ are i.i.d square integrable random 
    variables. Thus new information arrives on the market in discrete times given with a 
    Poisson process. At the arrival time of the new information, the efficient price moves.
    \item $\sum_{j=1}^{N_{t}^{\prime}} \theta\left(Q_{j}^{u}\right)\left(X_{j}-\rho X_{j-1}\right)$ where $(N^{\prime}_{t})_{t>0}$ is a Poisson process with intensity $\lambda_u$, $(Q_{j}^{u})_{j>0}$ the noise 
    traders volume distributions and $\theta$ a function that measures the size of noise traders trade impact, $X_j$ is the side of the $j^{th}$ trade of the noise traders(1 if noise trader's trade $j$ is a buyer initiated and -1 otherwise), $\rho$ is the first order auto correlation of the trade signs $X$. We assume that $(Q)$ and $(B)$ are independent.\\ 
    This part captures the efficient price move due to noise traders , particularly efficient price move
    due to the surprise component of the noise traders flow -the surprise component of the noise traders 
    flow is assumed to be $X_{j}-\rho X_{j-1}$ (a buy order followed by another buy order has less 
    surprise than a buy order followed by a sell order). We follow an here an 
    approach $a$ $la$ \protect\cite{madhavan1997security} to capture the toxicity due to the noise traders as 
    it is not really exact to assume it null.\\
    Let's define $\gamma=\mathbb{P}\left(X_{j}=X_{j-1}\right)$ and $\rho=\frac{\mathbb{E}\left(X_{j} X_{j-1}\right)}{\mathbb{V}\left(X_{i-1}\right)}$, we can easily draw that $\rho=2\gamma -1$. The case
    $\rho=0$ corresponds to independent trade signs whereas $\rho>0$ describes a positive auto-correlation of trade signs. \footnote{One critic to this modeling is that correlation decay exponentially $\mathbb{E}\left(X_{j} X_{j+k}\right)=\rho^{k}$ See \protect\cite{bouchaud2018trades, garcia2020multivariate} for a recent discussion}. We will assume that $\theta$ is a constant without loss of generality. One can make it dependent on the volume with the approach suggested in \protect\cite{madhavan1997security}: if we have k typical size ranges on can choose k parameters for each volume bucket but this will add complexity to the problem while the qualitative conclusions remain the same.
\end{itemize}

\subsection{Market participants}
We assume as discussed in the first part that there are four types of market participants:
\begin{itemize}
    \item One informed trader
    \item Informed market makers\\
    The informed market makers and the informed trader have superior information on the fair price 
    dynamics. They have low latency infrastructure\footnote{Having information first is not enough and an asymmetry can appear here between the IT and the IMM. C.A.Lehalle and P.Besson \protect\cite{lehalleclass} made an experiment showing that easy to process news cause price shifts without trades while difficult to process news cause price move with trades. In a word easy to process news allow informed market makers to update their price and thus is more beneficial for them while difficult to process news cause price moves with trades and so are more profitable to informed traders(one can imagine that IT can not value quickly with accuracy the news but know that on average following it is successful and thus react and adjust later while market makers have no incentives to cancel as they can make profit with traders misinterpreting the news). We will not consider this asymmetry in the following between the IT and the IMM.} and can exploit the market inefficiencies.
    They are able to assess efficient price jumps before other participants by using different techniques(statistics, news etc). We assume they receive the price jump just before it happens. Once there is a price jump, they compete to profit from this information. We introduce a parameter $f$ with $0\leq f \leq 1$ to measure the probability that the informed trader order is processed before the informed market maker one (so the probability that the informed market maker cancels his order is $1-f$). The introduction of the IMM and this parameter are the main differences and bring the main effects compared to the conclusions drawn by \protect\cite{huang2019glosten}
    \item One noise trader: he sends market orders without intelligence. Usually he is a liquidity taker trading for liquidity shocks unrelated to the actual dynamics. We assume that theses trades follow a compound Poisson process with intensity $\lambda_u$. We denote by $f_u$ the density of the noise trader volumes and $F_u$ their cumulative distribution function. 
    \item Noise market makers: We call them noise market makers as opposed to the informed market makers because we assume they don't receive the efficient price jump before it happens but right after it happens. We assume that they are risk neutral. We assume that they know the ratio of price jumps compared to the number of events happening in the market.
\end{itemize}

We denote by $L$ the cumulative shape function. $L(x)$ denote the available liquidity between prices $P(t)$ and $P(t)+x$ (We assume that there is no tick size for the moment and will relax this assumption later in the same way that \protect\cite{huang2019glosten}).

In the following parts we derive the results for the ask side of the book
(deriving them for the bid side can be done in a similar way).\\

Let's recall our main assumption: at a time of efficient price jump, the informed market maker sends an order to cancel his staying order in the limit order book on the side of the jump to avoid being adverse selected by the informed trader (one can imagine that he cancels and send a market order to adverse select the noise trader but for the framework this is the same than him canceling and the informed trader adverse selecting the noise trader\footnote{One should think to the ``self trade protection" mechanism in many exchanges that allow a market participant to not trade with himself}). t the same time the informed trader
sends a market order in a greedy way to hit all the available liquidity in the LOB between $P(t)$ and $P(t)+B$ where $B$ is the jump size (actually he sends marketable limit orders up to the level price B in order to avoid trading at a level price greater than B when the informed market makers has successfully canceled his order).
With probability $f$ the informed trader's order is processed before the informed market makers one. The sizable limit order quantity sent by the informed trader is $Q^{i}=L(B)$ where $L$ is as already said the shape of the book. $L$ is actually the sum of the books provided by the IMM and the NMM.

\subsection{Market Makers break-even condition and shape of the book (when $\theta = 0$)}
We derive the shape of the book in the case $\theta = 0$ corresponding to noise traders not impacting the efficient price (We relax this assumption in section \ref{subsec:toxicity}) and consider a general $f$. 
The particular case when $f = 1$ corresponds to the \protect\cite{huang2019glosten} case and we will compare market quality parameters in this particular case and the general case.

\subsubsection{The case without tick size}
We first consider the case without tick size. The general case with tick size is derived from this case in section \ref{subsubsec:ticksize}.
We define $G_{IMM}(x)$ and $G_{NMM}(x)$ as the \emph{conditional average profit} of a new infinitesimal order
if submitted at a price level x by respectively the informed and the noise market makers knowing that it is filled and without any information on the trade's initiator.

We start by computing the market makers expected gains.
\comment{\begin{definition}[Market makers expected gains] We define $G_{IMM}(x)$ and $G_{NMM}(x)$ as the \emph{conditional average profit} of a new infinitesimal order
if submitted at a price level x by respectively the informed and the noise market makers knowing that it is filled and without any information on the trade's initiator.
\emph{To derive $G_{IMM}(x)$ and $G_{NMM}(x)$ we define the following}:
\begin{enumerate}
    \item \emph{agg} a random variable that is equal to IT if the trade is initiated by the informed trader and NT if it is initiated by the noise trader
    \item \emph{$G_{IMM}^{NT}(x)$} the gain of the new infinitesimal order submitted by the IMM knowing that it is filled by the NT
    \item \emph{$G_{IMM}^{IT}(x)$} the gain of the new infinitesimal order submitted by the IMM knowing that it is filled by the IT
    \item \emph{$G_{NMM}^{NT}(x)$} the gain of the new infinitesimal order submitted by the NMM knowing that it is filled by the NT
    \item \emph{$G_{NMM}^{IT}(x)$} the gain of the new infinitesimal order submitted by the NMM knowing that it is filled by the IT
    \item \emph{$r = \frac{\lambda_i}{\lambda_{i} + \lambda_{u}}$} the ratio of efficient price jumps compared to all the events happening in the market
\end{enumerate}
\end{definition}}

\begin{proposition}\label{th:Gain_IMM}
The average profit of a new infinitesimal order if submitted at price level x by the IMM satisfies 

\begin{equation}
\label{eq:Gain_IMM}
G_{IMM}(x)=x-\frac{f r \mathbb{E}\left(B \mathbf{1}_{B>x}\right)}{(1-r) \mathbb{P}\left[Q^{u}>L(x)\right]+f r \mathbb{P}[B>x]}
\end{equation}
\end{proposition}

\begin{proposition}\label{th:Gain_NMM}
The average profit of a new infinitesimal order if submitted at price level x by the NMM satisfies 
\begin{equation}
\label{eq:Gain_NMM}
G_{NMM}(x)=x-\frac{r \mathbb{E}\left(B \mathbf{1}_{B>x}\right)}{(1-r) \mathbb{P}\left[Q^{u}>L(x)\right]+r \mathbb{P}[B>x]}
\end{equation}
\end{proposition}

\begin{remark}[Case $f=1$]
\comment{Note that the average profit for the NMM \eqref{eq:Gain_NMM} corresponds to the case $f=1$ for the IMM in \eqref{eq:Gain_IMM}.
In this case (f = 1), there is no difference between the IMM and the NMM as the IMM can not profit from his knowledge to cancel his order and thus suffers from the same adverse selection than the NMM.}
\end{remark}

\begin{proof}\label{proof:Gain_IMM}
We provide the proof in the case of the informed market maker, the case of the noise market maker can be derived in a similar way or with the remark just above. We have:
\begin{alignat}{10}
    G_{IMM}(x)=G_{IMM}^{NT}(x) \mathbb{P}(agg=NT \mid \text {Filled})+G_{IMM}^{IT}(x) \mathbb{P}(agg=IT \mid \text {Filled}) \notag \\
    G_{IMM}^{NT}(x)=x \emph{ and } G_{IMM}^{IT}(x)=x-\mathbb{E}(B \mid B>x) \notag \\
    \implies G_{IMM}(x)=x-\mathbb{E}(B\mid B>x) \mathbb{P}(agg=IT \mid Filled) \notag \\
    \implies G_{IMM}(x)=x-\frac{r f \mathbb{E}\left(B 1_{B>x}\right)}{\mathbb{P}(\text {Filled})} \notag
\end{alignat}
\emph{where we used conditional probability formula and the definition of conditional expectation}.

\begin{alignat}{10}
    \mathbb{P}(\text {Filled})=r \mathbb{P}(\text {Filled} \mid a g g=I T)+(1-r) \mathbb{P}(\text {Filled} \mid agg=NT) \notag \\
    \implies \mathbb{P}(\text {Filled}) = rf\mathbb{P}(B>x)+(1-r) \mathbb{P}\left(Q^{u}>L(x)\right) \notag 
\end{alignat}
\emph{which ends the proof.}\\
The main point in this proof is to keep in mind the fact that the IMM order on back of the queue being filled means that either there is an efficient price jump and he did not succeed to cancel before the IT market order's execution or the NT sends an order of size greater than $L(x)$
\end{proof}

We now discuss how the IMM and the NMM build the LOB with respect to these gains. Market makers compute the value of the shape functions in order to break even on average (conditional on being filled now). To do this the IMM and the NMM compute the respective values of the LOB shape function in order to have $G_{IMM}(x) = 0$ and $G_{NMM}(x) = 0$. From \eqref{eq:Gain_IMM} and \eqref{eq:Gain_NMM} one can see that they can not for the same shape have their gains equal to zero. Hence the following theorem.

\begin{theorem}[Cumulative LOB shape]
  \label{th:cumulative_lob_shape}
  \begin{enumerate}
      \item The cumulative LOB shape for the noise market makers is smaller than the one of the informed market makers.
      \item The effective (visible) LOB shape is thus imposed by the informed market makers and is given by
      \begin{equation}
      \label{eq:cumulative_lob_shape}
L(x)=F_{u}^{-1}\left[(1-f)+\frac{f}{1-r}-\frac{f r}{1-r} \mathbb{E}\left(\max \left(\frac{B}{x}, 1\right)\right)\right]
\end{equation}
    \item The shape function is decreasing with respect to f and with respect to r; it is increasing with respect to x.
  \end{enumerate}
\end{theorem}

\begin{proof}
Let's denote by $L_{i}(x)$ and $L_{u}(x)$ the optimal LOB shape for the IMM and the NMM (the LOB shape that allow them to break even on average).\\
From \eqref{eq:Gain_IMM} and \eqref{eq:Gain_NMM} it is straightforward to see that the break even condition correspond to $x=\frac{f r \mathbb{E}\left(B \mathbf{1}_{B>x}\right)}{(1-r) \mathbb{P}\left[Q^{u}>L_{i}(x)\right]+f r \mathbb{P}[B>x]}$ for the IMM and $x=\frac{ r \mathbb{E}\left(B \mathbf{1}_{B>x}\right)}{(1-r) \mathbb{P}\left[Q^{u}>L_{u}(x)\right]+ r \mathbb{P}[B>x]}$ for the NMM.\\
From the remark that $\mathbb{E}\left[\max \left(\frac{B}{x}, 1\right)\right]=\frac{\mathbb{E}\left(B 1_{B>x}\right)}{x}+\mathbb{P}(B<x)$, we then derive that:
  \begin{eqnarray}
  \label{eq:LOB_shape_informed}
    L_{i}(x)=F_{u}^{-1}\left[(1-f)+\frac{f}{1-r}-\frac{f r}{1-r} \mathbb{E}\left(\max \left(\frac{B}{x}, 1\right)\right)\right] \\
    \label{eq:LOB_shape_noise}
    L_{u}(x)=F_{u}^{-1}\left[\frac{1}{1-r}-\frac{r}{1-r} \mathbb{E}\left(\max \left(\frac{B}{x}, 1\right)\right)\right] 
  \end{eqnarray}
  Let's remark that $1-F_{u}\left(L_{i}(x)\right)=f\left(1-F_{u}\left(L_{u}(x)\right)\right)$ and as $0\leq f \leq 1$, we have $F_{u}\left(L_{i}(x)\right) \geq F_{u}\left(L_{u}(x)\right)$ and as $F_u$ is a non-decreasing function we have the first point of the theorem: 
  \begin{enumerate}
      \item $L_{i}(x) \geq L_{u}(x)$
      \item This point follows immediately from the precedent point.
      \item As $F_{u}^{-1}$ is a non decreasing function, it is sufficient to show that $$
h(f):=1-f+\frac{f}{1-r}-\frac{f r}{1-r} \mathbb{E}\left(\max \left(\frac{B}{x}, 1\right)\right)
$$ is a non-increasing function with respect to $f$. The derivatives of $h$ with respect to f is given by 
$$
h^{\prime}(f)=-1+\frac{1-r \mathbb{E}\left(\max \left(\frac{B}{x}, 1\right)\right)}{1-r}
$$ which is clearly non-positive.\\
The other statements can be derived in a similar way.
  \end{enumerate}
\end{proof}

Let's make some comments on this theorem. Each market maker computes its optimal LOB shape. When there is a possibility for a future profit, the market makers will add liquidity in the market until they are in their break even point with a vanishing gain. The interesting point is that while noise market makers will stop providing liquidity in order to avoid negative gains, informed one can still continue adding while having a non- negative gain when $f \neq = 1$. When the informed market makers become faster than the informed traders, that is when $f$ decreases the liquidity increases and and becomes eventually unbounded when $f \rightarrow 0$ (we recall that we are still in the case with no toxicity from noise traders).\\
We have seen that $x \rightarrow L(x)$ is an increasing function of x, we can derive the bid-ask spread based on this remark.

\begin{theorem}[Bid-ask spread]
  \label{th:bid_ask_spread_no_tick}
  \begin{enumerate}
      \item The cumulative LOB shape satisfies $L(x)=0$ for $0 \leq x \leq \phi$ and for $x > \phi$ it is increasing where $\phi$ is the unique solution of the following equation:
      \begin{equation}
      \label{eq:bid_ask_spread_no_tick}
\mathbb{E}\left[\max \left(\frac{B}{\phi}, 1\right)\right]=1+\frac{1}{2 f}\left(\frac{1}{r}-1\right)
\end{equation}
The bid-ask spread is equal to $2 \phi$.
      \item the bid-ask spread is increasing with respect to $f$ and to $r$.
  \end{enumerate}
\end{theorem}

\begin{proof}
We have already seen that the cumulative LOB shape is an increasing function of x, and as
market makers add volumes in order to break even on average. The condition to add volumes for the
informed market makers is then $L_{i}(x) > 0$ and for the noise market makers $L_{u}(x) > 0$.\\
As already noticed, informed market makers have the possibility to add volumes after the break even
condition of noise market makers and so determine the effective LOB shape. As $F_{u}^{-1}$ is increasing and
$F_{u}^{-1}(\frac{1}{2})=0$, we deduce the existence of $\phi$ such that $L_{i}(\phi) = 0$ with $\phi$ verifying the implicit equation \begin{equation}
\label{eq:phi_spread}
(1-f)+\frac{f}{1-r}-\frac{f r}{1-r} \mathbb{E}\left(\max \left(\frac{B}{\phi}, 1\right)\right)=\frac{1}{2}
\end{equation}
and for $x > \phi$ $L_{i}(x) > 0$. For $0 \leq x < \phi$ $L_i$ should be negative which means that market makers do not provide liquidity at these level prices and thus $L_{i}(x) = 0$ for these points. $\phi$ is thus the half bid-ask spread. It is straightforward to derive \ref{eq:bid_ask_spread_no_tick} from \ref{eq:phi_spread}.\\
The second point is straightforward from the equation solved by the half bid-ask spread.
\end{proof}

\begin{remark}
 Actually the theorem tells us that the bid-ask spread is fixed by the informed market makers as the spread for noise market makers (case $f=1$) denoted by $\mu$ solves $\mathbb{E}\left[\max \left(\frac{B}{\mu}, 1\right)\right]=1+\frac{1}{2}\left(\frac{1}{r}-1\right)$ and consequently $\phi < \mu$. \emph{The meaning is clear: after noise market makers stop adding volumes and fix their spread, the informed
market makers still continue to add volumes near the efficient price. We remark (as expected) that the spread does not depend on noise traders distributions as we have assumed that they do not have any toxicity for market makers. We recall that we do not make any assumption on inventory costs (or any kind of costs due to liquidity provision) that would let emerge the spread. The spread appears as a purely mechanism to balance asymmetry of information between market makers and informed traders.}
 
\end{remark}

\subsubsection{The non-zero tick size case}
\label{subsubsec:ticksize}
In this section we study the effect of introducing a tick size in the same way as \protect\cite{huang2019glosten}. We suppose that we have the same efficient price than in the case without tick size. The cumulative LOB shape becomes now a piece-wise constant function as liquidity can be add only on the prices in the tick grid. Let's define $d = P_{alpha}(t) - P(t)$ where we define $P_{alpha}(t)$ as the smallest admissible price in the tick grid that is greater or equal to the current efficient price $P(t)$.

We now consider the cumulative LOB shape $L(x)$ as defined on the tick grid by $L^{d}(i)$ where \(i \in \mathbb{N}^{*}\) is the $i^{th}$ closest price to $P(t)$ in the ask book (we consider as in the previous section the ask side of the book, the derivation for the bid size is straightforward). We define $$
L^{d}(i)=L(d+(i-1) \alpha)
$$ 
where $\alpha$ is the tick size. The quantity placed at the $i^{th}$ limit is defined by $$
l^{d}(i)=L^{d}(i)-L^{d}(i-1)
$$
We still assume that as soon as there is an efficient price jump, the informed trader and the informed market maker make a race to benefit from this information. We keep the same parameter $f$.
The informed trader sends a marketable limit order of size $Q^{i} = L^{d}(i)$ if $B \in[d+(i-1) \alpha, d+i \alpha]$ and the informed market maker sends an order to cancel all his waiting orders at this limit.

The computation of the informed market makers and noise market makers expected gains conditional on being filled now is quasi-exactly the same than in the case without tick size. We denote by $G_{IMM}^{d}(i)$ and $G_{NMM}^{d}(i)$ the conditional gains of a new infinitesimal passive order placed at the $i^{th}$ level for the IMM and the NMM.

\begin{proposition}\label{th:Gain_IMM_tick}
The average gain of a new infinitesimal order submitted by the informed market maker at level price $i$ conditional on being filled now is
\begin{equation}
G_{I M M}^{d}(i)=G_{I M M}(d+(i-1) \alpha)=d+(i-1) \alpha-\frac{f r \mathbb{E}\left(B 1_{B>d+(i-1) \alpha}\right)}{(1-r) \mathbb{P}\left[Q^{u}>L^{d}(i)\right]+f r \mathbb{P}[B>d+(i-1) \alpha]}
\end{equation}
\end{proposition}

\begin{proposition}\label{th:Gain_NMM_tick}
The average gain of a new infinitesimal order submitted by the noise market maker at level price $i$ conditional on being filled now is
\begin{equation}
G_{N M M}^{d}(i)=G_{N M M}(d+(i-1) \alpha)=d+(i-1) \alpha-\frac{ r \mathbb{E}\left(B 1_{B>d+(i-1) \alpha}\right)}{(1-r) \mathbb{P}\left[Q^{u}>L^{d}(i)\right]+ r \mathbb{P}[B>d+(i-1) \alpha]}
\end{equation}
\end{proposition}

\begin{theorem}[Cumulative LOB shape with tick size]
  \label{th:cumulative_lob_shape_tick_size}
  \begin{enumerate}
      \item The cumulative LOB shape for the noise market makers is smaller than the one of the informed market makers.
      \item The effective (visible) LOB shape is thus imposed by the informed market makers and is given by
      \begin{equation}
L^{d}(i)=F_{u}^{-1}\left[(1-f)+\frac{f}{1-r}-\frac{f r}{1-r} \mathbb{E}\left(\max \left(\frac{B}{d+(i-1) \alpha}, 1\right)\right)\right]
\end{equation}
    \item The shape function is decreasing with respect to f and with respect to r (all things being equal)
  \end{enumerate}
\end{theorem}

The proof is identical to the case without tick size as we still assume market makers add orders to break even on average. The comments for theorem \eqref{th:cumulative_lob_shape} still hold. In a word, when noise market makers stop providing liquidity (otherwise their gain becomes negative), the informed market makers can still continue to provide liquidity while having a non-negative gain. When the informed market makers become infinitely faster than the informed traders, the liquidity becomes unbounded. In any case an interesting point to notice from this framework is that the cumulative LOB shape is strictly increasing and goes to infinity with $i$. So the book is never empty and the noise trader can always find volume to trade with. This is due to the smoothness hypotheses we made on the noise traders volume and on the efficient price distributions.

Now our goal is to derive the spread when the tick size is not zero.

\begin{theorem}
  \begin{enumerate}
      \item The LOB shape satisifies $l^{d}(i) = 0$ for $0 < i < k_d$ where $k_d$ is determined by the following equation:
      \begin{equation}
      \label{eq:kd}
k_{d}=1+\lceil\frac{\phi-d}{\alpha}\rceil
\end{equation}
where $\lceil x \rceil$ denotes the smallest integer that is larger than $x$.
  \end{enumerate}
  \item The bid ask spread is given by
  \begin{equation}
  \label{eq:spread_tick}
\phi_{\alpha}^{d}=\alpha\left(\lceil\frac{\phi-d}{\alpha}\rceil+\lceil\frac{\phi+d}{\alpha}\rceil\right)
\end{equation}
\end{theorem}

\begin{proof}
We showed in the case the tick size is zero that there exists $\phi$ such that for $x < \phi$ $L(x)=0$ and for $x> \phi$ $L(x) > 0$. This result remains true in the discrete LOB for $k_d$ where $k_{d}=\min \{k \in \mathbb{N}^{*} \mid d+(k-1) \alpha>\phi\}$, so we have \eqref{eq:kd}. The other results follow by a symmetry argument for the bid side.
\end{proof}

\subsection{Adding toxicity from noise traders (the case $\theta \neq 0$)}
\label{subsec:toxicity}
We come back to the general case with a non zero constant $\theta$. we consider there is no tick size as the
extension can be made in a very similar way than previously. we still consider the ask side of the book.

\begin{proposition}\label{th:Gain_IMM_toxicity}
The average profit of a new infinitesimal order if submitted at price level x by the IMM satisfies
\begin{equation}
\label{eq:Gain_IMM_toxicity}
G_{I M M}(x)=x-\Theta+\left(\Theta-\frac{\mathbb{E}\left(B 1_{B>x}\right)}{\mathbb{P}(B>x)}\right) \frac{r f \mathbb{P}(B>x)}{(1-r) \mathbb{P}\left(Q^{u}>L(x)\right)+r f \mathbb{P}(B>x)}
\end{equation} where \(\Theta:=\theta \mathbb{E}\left(X_{j}-\rho X_{j-1} \mid X_{j}=1\right)\)
\end{proposition}

\begin{proposition}\label{th:Gain_NMM_toxicity}
The average profit of a new infinitesimal order if submitted at price level x by the NMM satisfies
\begin{equation}
G_{N M M}(x)=x-\Theta+\left(\Theta-\frac{\mathbb{E}\left(B 1_{B>x}\right)}{\mathbb{P}(B>x)}\right) \frac{r \mathbb{P}(B>x)}{(1-r) \mathbb{P}\left(Q^{u}>L(x)\right)+r \mathbb{P}(B>x)}
\end{equation} where \(\Theta:=\theta \mathbb{E}\left(X_{j}-\rho X_{j-1} \mid X_{j}=1\right)\)
\end{proposition}

\begin{proof}
We provide the proof in the case of the informed market maker, the case of the noise market maker can be derived in a similar way. We have:
\begin{alignat}{10}
    G_{IMM}(x)=G_{IMM}^{NT}(x) \mathbb{P}(agg=NT \mid \text {Filled})+G_{IMM}^{IT}(x) \mathbb{P}(agg=IT \mid \text {Filled}) \notag \\
    G_{IMM}^{NT}(x)=x-\theta \mathbb{E}\left(X_{j}-\rho X_{j-1} \mid X_{j}=1\right)\   \emph{ and } G_{IMM}^{IT}(x)=x-\mathbb{E}(B \mid B>x) \notag 
\end{alignat}
\emph{The rest of the proof is as in \eqref{proof:Gain_IMM}}.
\end{proof}
We now give the shape of the book and the bid-ask spread.

\begin{theorem}[Cumulative LOB shape with $\theta \neq 0$]
  \label{th:cumulative_lob_shape_toxic}
  \begin{enumerate}
      \item The cumulative LOB shape for the noise market makers is smaller than the one of the informed market makers.
      \item The effective (visible) LOB shape is thus imposed by the informed market makers and is given by
\begin{equation}
L(x)=F_{u}^{-1}\left[1+\frac{r f}{1-r} \frac{x}{x-\Theta}-\frac{f r}{1-r} \frac{x}{x-\Theta} \mathbb{E}\left(\max \left(\frac{B}{x}, 1\right)\right)\right]
\end{equation}
    \item The shape function is decreasing with respect to f and with respect to r; it is increasing with respect to x.
     \item The cumulative LOB shape satisfies $L(x)=0$ for $0 \leq x \leq \phi_{\theta}$ and for $x > \phi_{\theta}$ it is increasing where $\phi_{\theta}$ is the unique solution of the following equation:
     \begin{equation}
     \label{eq:spread_toxicity}
\mathbb{E}\left[\max \left(\frac{B}{\phi_{\theta}}, 1\right)\right]=1+\frac{1}{2 r f}(1-r) \frac{\phi_{\theta}-\Theta}{\phi_{\theta}}
\end{equation}
  \end{enumerate}
\end{theorem}
The proof can be derived in a very similar way than the case $\theta = 0$ with the gain in the two previous propositions \eqref{th:Gain_IMM_toxicity} and \eqref{th:Gain_NMM_toxicity}.
The first remark to make is that as expected $\phi_{\theta}>\Theta$ as the left term in \eqref{eq:spread_toxicity} is greater than 1 and in order to have the right term greater than 1 it is necessary to have $\phi_{\theta}>\Theta$.

The second point to remark is that the spread in this case is greater than the spread in the case $\theta = 0$ as the right part of \eqref{eq:spread_toxicity} is decreasing with $\Theta$ and the minimum is for $\theta = 0$.

The model tells us that adding toxicity from the noise traders just increase the bid-ask spread. Choosing a general impact function in this framework does not give a closed formula for the LOB shape but an implicit equation between the different parameters of the model. The qualitative conclusions do not change but to not overcharge this paper we decide to not add it. The link to the discrete LOB can be derived without difficulty as done in the previous part.

\section{Extension of the model to heterogeneously informed market makers}
\label{sec:extension}

Our goal now is to extend the baseline model in the case there are many sources of jumps in the efficient price and for each kind of jump there are specialized informed market makers that can only capture this kind of jumps. We still assume that there is a race for order insertion between informed market makers and informed traders in the book as soon as there is jump. Our goal is to assess the impact of having competition between many market makers with different type of signals. We will assume in the following that there is no toxicity from noise traders ($i.e$ $\theta = 0$) as we have seen that adding this does not fundamentally change the conclusions drawn.

Let's denote by $n$ the number of different sources of jumps (so the number of different types of informed market makers\footnote{one can think about it as specialized types of market makers: for example some market makers have very good models to capture information with lead-lag, others to process news, etc}). We will start by the case $n = 2$ and derive the explicit formula before giving the formula for the general $n$-case.

We consider an asset whose efficient price is given by
$$
P(t)=P_{0}+\sum_{l=1}^{N_{t}} B_{l}^{0}+\sum_{j=1}^{N_{t}^{\prime}} B_{j}^{1}
$$
where $\sum_{l=1}^{N_{t}} B_{l}^{0}$ (resp $\sum_{l=1}^{N_{t}^{\prime}} B_{l}^{1}$) is a compound Poisson process and $(N_{t})_{t \geq 0}$ (resp $(N_{t}^{\prime})_{t \geq 0}$) is a Poisson process with intensity $\lambda_{i}^{0}$ (resp $\lambda_{i}^{1}$) and the $(B_{l}^{0})_{l>0}$ (resp $(B_{l}^{1})_{l>0}$) are $i.i.d$ square integrable random variables. We suppose there are two types of informed traders (IT0 and IT1) and two types of market makers (IMM0 and IMM1). We recall that we assume that IT0 and IMM0 (resp IT1 and IMM1) can only capture the jumps of $(B_{l}^{0})_{l>0}$ (resp $(B_{l}^{1})_{l>0}$). We keep our race assumption between IT0 and IMM0 (resp IT1 and IMM1) at a time of efficient price jump with a race parameter $f$ that we assume to be the same for all. We still assume the presence of noise traders whose trade volumes are denoted as previously by $(Q_{j}^{u})_{j>0}$.

We define $r_{0}:=\frac{\lambda_{i}^{0}}{\lambda_{i}^{1}+\lambda_{i}^{0}+\lambda_{u}}$ and $r_{1}:=\frac{\lambda_{i}^{1}}{\lambda_{i}^{1}+\lambda_{i}^{0}+\lambda_{u}}$

We will derive the gain of a new infinitesimal order added by the market makers conditional on
being filled right now.

\begin{proposition}\label{th:Gain_IMM_MULTISOURCES}
For $j \in \{0,1\}$ he gain of a new infinitesimal order submitted at price level x by the
  informed market maker $j$ conditional on being filled is given by
  \begin{equation}
G_{I M M, j}(x)=x-\frac{r_{j} f \mathbb{E}\left(B 1_{B>x}\right)}{\left(1-r_{0}-r_{1}\right) \mathbb{P}\left(Q^{u}>L(x)\right)+r_{j} f \mathbb{P}\left(B^{j}>x\right)+r_{1-j} \mathbb{P}\left(B^{1-j}>x\right)}
\end{equation}
\end{proposition}

\begin{proof}
Without loss of generality we give the proof for IMM1.
\begin{alignat}{10}
    G_{IMM1}(x)=G_{IMM1}^{agg=NT}(x)[1-\mathbb{P}(agg=IT1 \mid \text {Filled})-\mathbb{P}(agg=IT0 \mid \text {Filled})]\notag \\
    +G_{I M M 1}^{a g g=I T 1}(x)\mathbb{P}(agg=IT1 \mid \text {Filled}) \notag \\
    +G_{I M M 1}^{a g g=I T 0}(x) \mathbb{P}(a g g=I T 0 \mid \text {Filled}) \notag \\
    \mathbb{P}(agg=IT1 \mid \text {Filled})=\frac{\mathbb{P}(agg=IT1, \text {Filled})}{\mathbb{P}(\text {Filled})}=\frac{r_{1} f \mathbb{P}\left(B^{1}>x\right)}{\left(1-r_{1}-r_{0}\right) \mathbb{P}\left(Q^{u}>L(x)\right)+r_{1} f \mathbb{P}\left(B^{1}>x\right)+r_{0} \mathbb{P}\left(B^{0}>x\right)} \notag \\
    \mathbb{P}(agg=IT0 \mid \text {Filled})=\frac{\mathbb{P}(agg=IT0, \text {Filled})}{\mathbb{P}(\text {Filled})}=\frac{r_{0} \mathbb{P}\left(B^{1}>x\right)}{\left(1-r_{1}-r_{0}\right) \mathbb{P}\left(Q^{u}>L(x)\right)+r_{1} \mathbb{P}\left(B^{1}>x\right)+r_{0} \mathbb{P}\left(B^{0}>x\right)} \notag \\
    G_{I M M 1}^{a g g=NT}(x)=x \emph{ , }
    G_{I M M 1}^{a g g=I T 1}(x)=x-\mathbb{E}\left(B^{1} \mid B^{1}>x\right) \emph{ and }
    G_{I M M 1}^{a g g=I T 0}(x)=x-\mathbb{E}\left(B^{0} \mid B^{0}>x\right) \notag
\end{alignat}
\emph{From this it is easy to get the result of the proposition}.
\end{proof}
To derive the LOB shape let's remark that each of the three types of market makers will derive his optimal shape based on his gain. The main results of the first part do not change. The informed market makers will put a supplementary volume compare to the noise market makers. We give in the following theorem the cumulative LOB shape for each informed market maker.

\begin{theorem}\label{th:cumulative_lob_shape_2_sources}
  \begin{enumerate}
      \item For $j \in \{0,1\}$ he optimal cumulative LOB shape for informed market maker $j$ is given by 
      \begin{equation}
L_{j}(x)=F_{u}^{-1}\left[1+\frac{r_{j} f+r_{1-j}}{1-r_{1}-r_{0}}-\frac{1}{1-r_{0}-r_{1}}\left(r_{j} f \mathbb{E}\left(\max \left(\frac{B^{j}}{x}, 1\right)\right)+r_{1-j} \mathbb{E}\left(\max \left(\frac{B^{1-j}}{x}, 1\right)\right)\right]\right.
\end{equation}
    \item The effective shape of the book is given by 
    \begin{equation}
L(x)=F_{u}^{-1}\left[\max \left(F_{u}\left(L_{o}(x), F_{u}\left(L_{1}(x)\right)\right]\right.\right.
\end{equation}
    \item When $r_{0}r_{1} > 0$ and when $f = 0$ the shape function is not infinite as in the case with only one informed market maker.
  \end{enumerate}
\end{theorem}

The proof of the theorem is straightforward using the previous proposition and the break even condition for the different market makers.

The interesting point when considering many market makers is that a monopolistic situation is optimal for liquidity. Indeed in the case with just one informed market maker, we saw that when $f \rightarrow 0$ the liquidity became unbounded while when there is competition between the market makers having $f$ near 0 does not allow to have unbounded liquidity (all other parameters being fixed) This is in a sense logical. When there is just on informed market maker, he is willing to provide a smaller spread when a trader (noise) requests quotes, allowing a better market quality (smaller spreads and greater volumes) for liquidity traders.

Let's now give the general formula when there are $n$ types of informed market makers capturing each one a particular kind of jump. We assume that all the market makers and informed traders share the same common $f$.

\begin{theorem}\label{th:cumulative_lob_shape_n_sources}
  \begin{enumerate}
      \item For $k \in \mathbb{I} 0, n-1 \mathbb{J}$ the optimal cumulative LOB shape for informed market maker k is given by
      \begin{equation}
L_{k}(x)=F_{u}^{-1}\left[1+\frac{r_{k} f+\sum_{j \neq k} r_{j}}{1-\sum_{j} r_{j}}-\frac{1}{1-\sum_{j} r_{j}}\left(r_{k} f \mathbb{E}\left(\max \left(\frac{B^{k}}{x}, 1\right)\right)+\sum_{j \neq k} r_{j} \mathbb{E}\left(\max \left(\frac{B^{j}}{x}, 1\right)\right)\right]\right.
\end{equation}

    \item The effective shape of the book is given by 
    $$
L(x)=F_{u}^{-1} \left[ \max _{k}\left(F_{u}\left(L_{k}(x)\right)\right) \right].
$$
    \item When $f \rightarrow 0$ and $\exists k \neq j ,r_{k} r_{j}>0$ the shape function is not infinite as in the case with only one informed market maker.
  \end{enumerate}
\end{theorem}

The proof is straightforward by computing the gain of an infinitesimal order added by an informed market maker conditional on a fill and using the break even condition.
The results induced by this theorem are in line with the results discussed in the case with just two informed market makers.
We could use a more realistic modeling by allowing informed market makers to capture more than one kind of jump (and introduce different race parameters between to have a more realistic framework) but the conclusions would essentially remain the same as long as we do not suppose that there is one perfectly informed market maker that can access all the efficient price jumps.

\section{Numerical examples of LOB shape}
We give some numerical examples of LOB (ask side) for different parameters $r$, $d$ and $f$. We take a normal distribution for noise traders trade volumes with standard deviation 10 and a Pareto distribution for the absolute value of the efficient price jumps with shape 3 and scale 0.005. We take a tick size of 0.01.

\begin{figure}[H]
\centering
 \includegraphics[width=\textwidth]{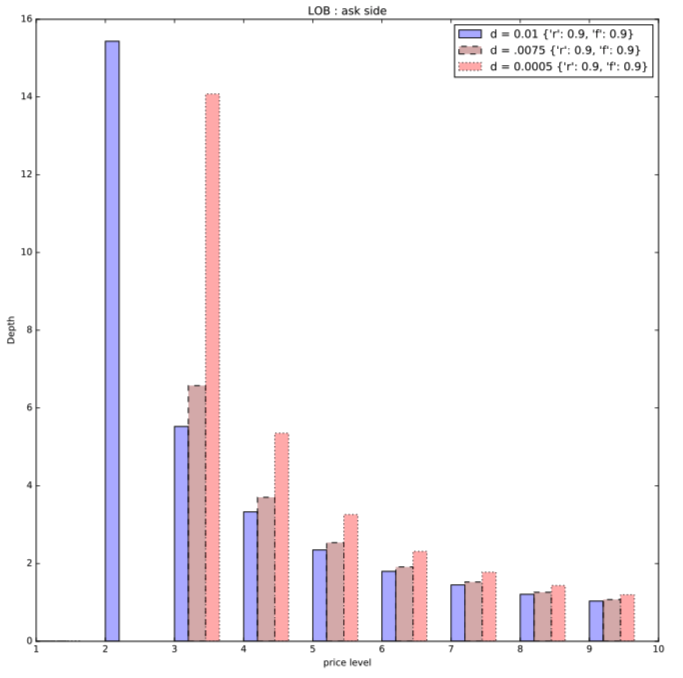}
 \caption{\textbf{LOB shapes with the corresponding parameters $r=0.9$, $f=90$ and different values of $d$}\newline
We take for the noise traders volume distributions a normal distribution with an arbitrary standard deviation of 10 and for the efficient price jumps a Pareto distribution with shape parameter of 3 and scale parameter of 0.005. We take a tick size of $0.01$.}
 \label{fig:lobr90f90}
\end{figure}

\begin{figure}[H]
\centering
 \includegraphics[scale=0.8]{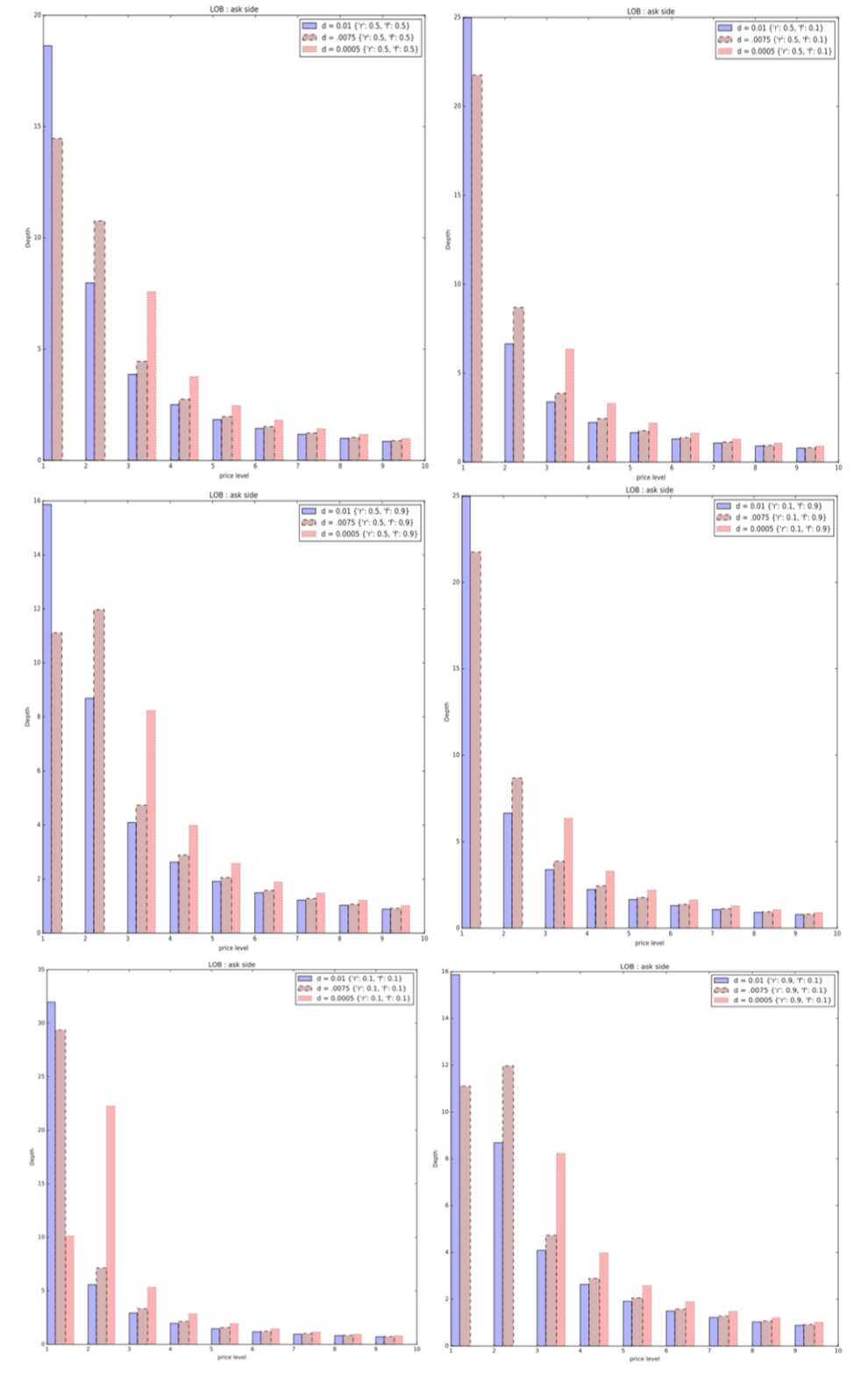}
 \caption{\textbf{LOB shapes with the corresponding parameters $r$,$f$,$d$}}
 \label{fig:lob}
\end{figure}

\section{Some foundations of the framework from empirical market facts}
We will discuss the foundations of our modeling giving some insights on the different facts justifying the different choices we made. Let's start by our assumption on the efficient price dynamics. One can think about the more general case with correlated assets. Let's assume we have two assets sharing some common dynamics with two corresponding books. In book 1 we assume there are NMM1 (Noise Market Makers 1), NT1 ( Noise Trader 1) and in book 2 NMM2 and NT2. We assume there are IMM and IT that act in the two books (the same)-the idea behind being that informed market makers in one book can act as informed traders in the other as soon as they get executed and get insight on the efficient price. Namely we assume that the dynamics of the two assets are given by
\begin{eqnarray}
    P_{t}^{1}-P_{0}^{1}=\sum_{t_{i}^{1} \leq t} \epsilon_{t_{i}^{1}}+\sum_{t_{i}^{2} \leq t} B_{t_{i}^{2}}+\sum_{t_{i}^{3} \leq t} \theta^{u, 1}\left(Q_{t_{i}^{3}}^{u, 1}\right)+\beta \sum_{t_{i}^{4} \leq t} \theta^{u, 2}\left(Q_{t_{i}^{4}}^{u, 2}\right) \notag \\
    P_{t}^{2}-P_{0}^{2}=\sum_{t_{i}^{5} \leq t} \epsilon_{t_{i}^{5}}+\sum_{t_{i}^{2} \leq t} B_{t_{i}^{2}}+\sum_{t_{i}^{4} \leq t} \theta^{u, 2}\left(Q_{t_{i}^{4}}^{u, 2}\right)+\beta \sum_{t_{i}^{3} \leq t} \theta^{u, 1}\left(Q_{t_{i}^{3}}^{u, 1}\right) \notag
\end{eqnarray}
The first term in each equation is the idiosyncratic component of each asset (that does not impact the efficient price of the other asset). The second term is the same for the two and represents the fundamental common driver of the two assets (one can think on the way news affect two futures in the same underlying with different maturities). The third term of each asset represents how the noise trader impacts the efficient price and the last term captures how the efficient price is influenced by trades in the other asset with $\beta$ a coefficient between 0 and 1. One can see $beta$ as the cross-impact\footnote{Empirical studies show evidence of considerable cross-impact effects in stock markets \protect\cite{hasbrouck2001common,pasquariello2015strategic,wang2016cross, benzaquen2017dissecting}} coefficient of the 2 assets independently of fundamental jumps. It captures for example in the case of the same asset in two exchanges how liquidity traders route their orders in different venues. Note however that the true correlation of the two assets is not $\beta$ but depend on the different jumps parameters and $2^{nd}$-order moments of the different terms. This case is general as it can take into account the different situations in which market participants trade: different assets with common underlying or the same asset traded in different exchanges. Thus our framework discussed in this paper takes into account this more general case and it is straightforward to derive the LOB shape of the two books. Notice that this framework enforces the existence of many kind of efficient price jumps but this is actually relevant with respect to what happens in markets. For example market makers can receive confirmation of trades in which they are involved before it is publicly displayed as the channels in which these information pass are different.\\

One should still not forget that LOBs drawn are ``static" LOBs. All the gains we computed are conditional on having a trade right now. Since when there is a large queue, the life of the order will not end with the next trade (and traders will not cancel and resubmit their limit orders after every single trade) the order will probably move in the queue if not executed by the next trade. As a result a dynamical model is needed to take into account the dynamics of the queue. We let this for further works or refer to \cite{moallemi2014value} for a model capturing the dynamical part of the queue position value (the authors solve this just for the best limit assuming this limit does not move which is restricting). We let this direction for further research. \\

Let's now discuss the motivations in introducing the race parameters between informed market makers and informed traders. We recall that we introduce a parameter $f$ , $0 \leq f leq 1$ to measure the probability that the informed trader order is processed before the informed market maker one (so the probability that the informed market maker cancels his order is $1-f$). The introduction of this parameter is one of the key point of the model. This parameter makes actually sense in most of the exchanges as there is a non-zero probability that 2 orders sent quasi-simultaneously are executed in the opposite sense. On CME exchange-which is one of the most liquid exchanges for futures and options in the world- for example for two orders sent in a delta time of less than $10ns$ the probability the last sent order is executed before the first one is $0.5$ (the $10ns$ delta time is an approximation of the time necessary for $100$ bits to reach the switches in the L1 layer of CME GLink. This so-called `Glink 10G" uses a $10 Gbps$ debit and in order to transmit $100$ bits the necessary time is roughly $10ns$ by a simple calculation). Another interpretation of the $f$ parameter is the randomizer some exchanges added to resort the different incoming orders or the so-called `speed bumps" that allow market makers to cancel their orders before the order sent by liquidity takers are executed.

\section{Conclusion and prospects}
In this paper, we assess the existing literature considering one generic group of market makers. We bring empirical evidence showing that we can classify market participants based on different parameters such as the moment they insert their orders, the relative volume they trade with respect to the available volume at the best limit, the number of updates of their orders etc.

We introduce an agent-based model for the LOB. Built on the Glosten-Milgrom approach and the recent Huang-Rosenbaum-Saliba, we use a zero-profit condition for the different types of market makers which enables us to derive a link between the proportion of different types of traders, the race parameter between informed market makers and informed traders, and the LOB shape. We discuss the effect of introducing a tick size.

We then extend the model to take into account toxicity from noise traders auto-correlated trades and competition between different types of informed market makers. We discuss how taking into account the noise traders trades signs auto-correlation would increase the bid-ask spread by a mechanical part due to this auto-correlation.

We discuss also how the race parameter could impact the liquidity and the bid-ask spread but did not discuss the interesting retro-effects of this race parameter that is the following: when all liquidity providers cancel simultaneously their orders, some liquidity deterioration will be caused.

We then show the model is realistic with respect to market participants motivations to trade with a discussion on how a `multi assets" model is straightforward to plug-in the framework.

We bring some justifications to the race parameter we introduce and relate it to `speed-bumps". Nevertheless with a speed-bump, an informed trader who has access to private information will likely be filed $100\%$ of the time\footnote{`Trading with a noise trader" s a private information in many real words cases: bilateral trading or any exchange which is either \textbf{non transparent} -Reuters and EBS where prices are sent periodically in a high-low fashion without volumes- or \textbf{delayed} -some dark pools where trade reporting is delayed.}. We do not discuss the feedback effects of speed-bumps which bring market makers to become incentivized to traders speeds increment, reducing the marginal cost of getting faster for HFT and consequently failing in their intended purpose of protecting market makers. We refer to \protect\cite{aoyagi2018strategic} for a more detailed description of this effect.\\

One of the main limitations of our model is that it does not capture well the market feasibility that is how a trade can occur. A market is not admissible if no one wants to be the trader or the market maker: a trader will likely not trade against a market maker if he is aware that the market maker is informed and vice versa . Indeed we do not address the mechanical substitution that would occur between traders and market makers: who would not be the informed market maker in this context ?

We do not discuss the intricate question on how inefficiencies n one market are linked to the price discovery in other markets, in particular we do not address the accuracy and transparency of data provided by intermediaries markets -which is one of the goals of regulators. In the US for example while the regulation NMS pushes to fragmentation of order flow, the improvement effect on price discovery is not really clear. Some studies show fragmentation into the dark pools raises effective spread and increases price manipulation \protect\cite{aitken2015fragmentation, chung2012regulation}. In Europe the Paris-based regulators group ESMA (European Securities and Markets Authority) aims for a regulatory regime with more transparency with the recent MiFID II -Markets in Financial Instruments Directive II- that is promoting the adoption of more transparent system and has for example capped the percentage of shares for a given product that can be traded in dark pools. In this context where regulators seem to care about the efficiency of the price discovery, adding speed-bumps would bring a supplementary loss of transparency.

Finally in our approach, we do not consider any inventory management effect and in practice it is a key point that impact liquidity (in particular when there are more informed traders than noise traders). Extending the framework to take into account inventory management is challenging but would bring more insight in the effects of race for insertion and inventory management on liquidity and price discovery.

\section*{Acknowledgments}

The author warmly thanks Cothereau A., Hattersley M. and Laffitte P. for their fruitful mentoring and discussions.

\appendix
\renewcommand*{\thesection}{\Alph{section}}
\section{Classification of market participants}
\label{app:graphs}

\begin{figure}[H]
\centering
 \includegraphics[scale=0.6]{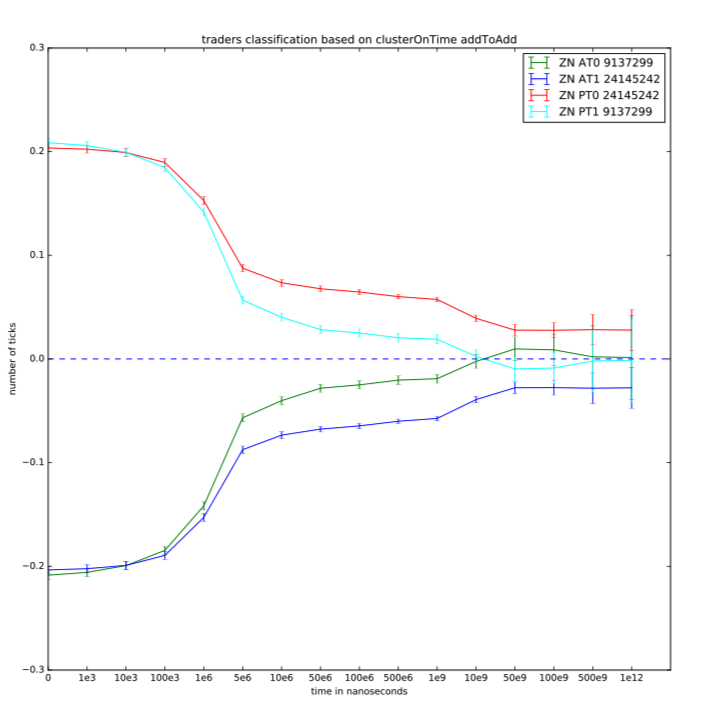}
 \caption{\textbf{ZN trade signature based on add to add duration with a threshold of $10^{7}ns$}}
 \label{fig:sigAddToAdd}
\end{figure}

\begin{figure}[H]
\centering
 \includegraphics[scale=0.3]{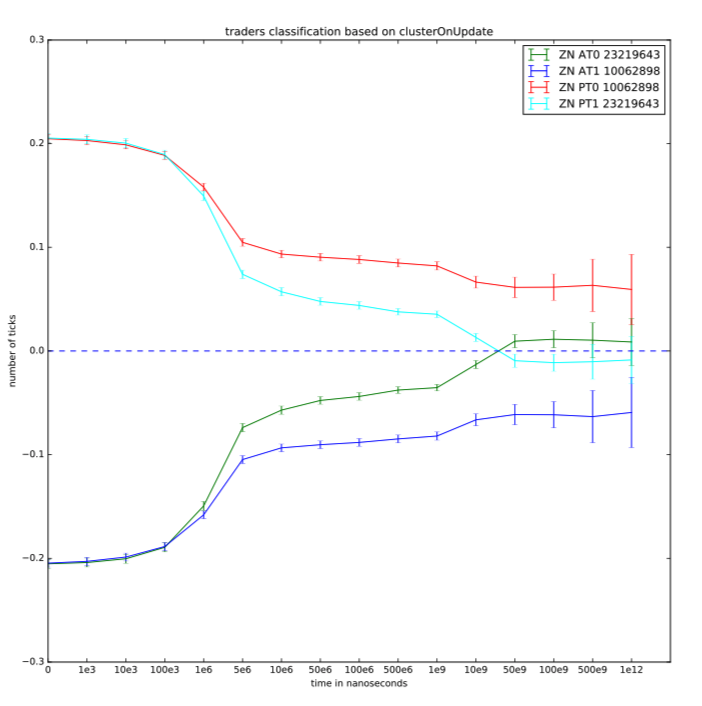}
 \caption{\textbf{ZN trade signature based on update existence between order insertion and order execution}}
 \label{fig:sigUpdate}
\end{figure}

\begin{figure}[H]
\centering
 \includegraphics[scale=0.3]{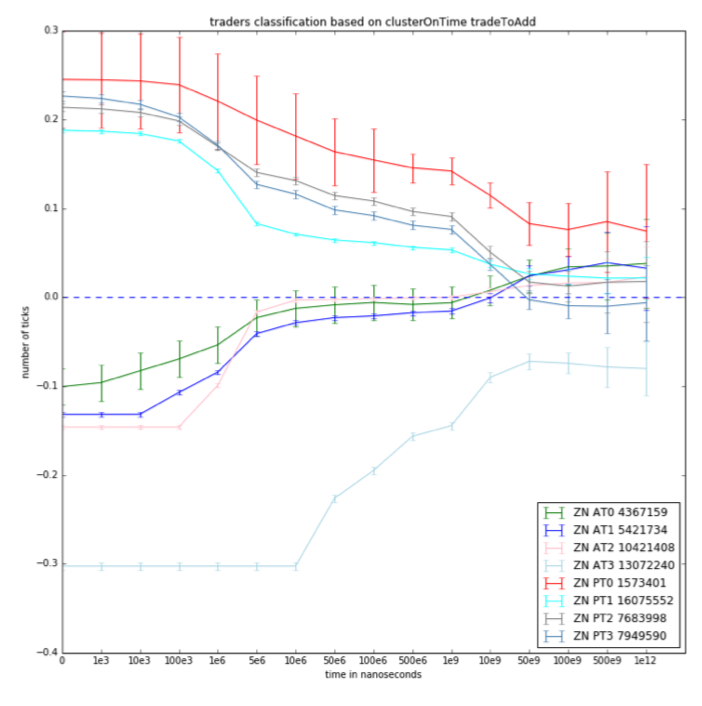}
 \caption{\textbf{ZN trade signature based on trade to add duration with thresholds $(10^{4}, 10^{7}, 10^{9})$}}
 \label{fig:sigTradeToAddMulti}
\end{figure}

\begin{figure}[H]
\centering
 \includegraphics[scale=0.5]{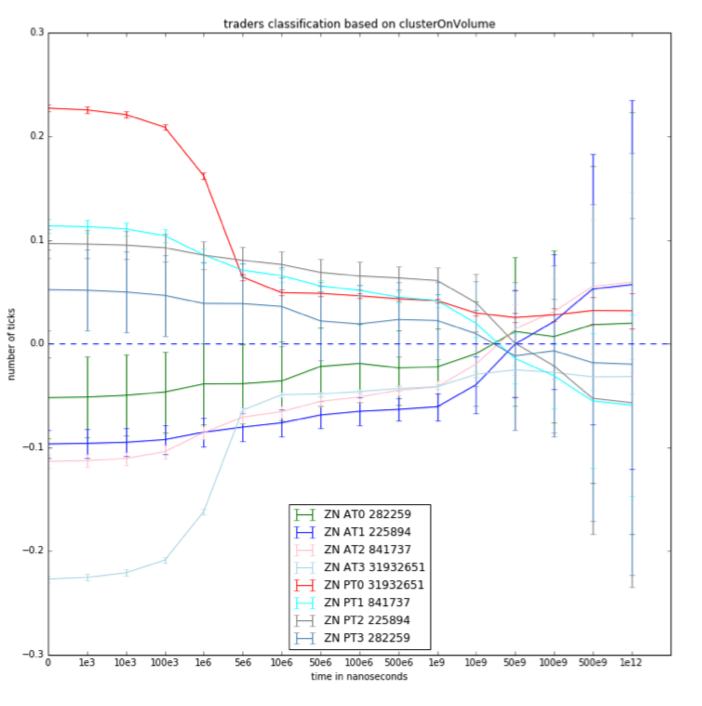}
 \caption{\textbf{ZN trade signature based ration volumes traded and available volume at best limit at trade time with thresholds $(0.25, 0.5, 0.75)$}}
 \label{fig:sigVolumeMulti}
\end{figure}

\clearpage

\printbibliography


\end{document}